\documentclass{fundam}

\usepackage{enumerate}
\usepackage{url}
\usepackage{amsmath}
\usepackage{mathtools}
\usepackage{cite}
\usepackage{amssymb}
\let\emptyset\varnothing
\usepackage{graphicx}
\usepackage{caption}
\usepackage{pdfpages}
\usepackage{adjustbox}
\usepackage{rotating}
\usepackage{floatpag}
\usepackage{wrapfig}
\usepackage[T1]{fontenc}
\usepackage{mathrsfs}
\usepackage{multirow}
\usepackage{hhline}
\usepackage{enumerate}
\usepackage{array}
\usepackage{adjustbox}

\usepackage[algo2e,vlined,ruled,commentsnumbered]{algorithm2e}
\usepackage{algpseudocode}
\usepackage{framed}
\usepackage[english]{babel}

\renewcommand{\thesection}{\Roman{section}}

\newtheorem {defnn}{Definition}
\newtheorem {condition}{Condition}
\newenvironment{reason}{\noindent {\bf Reason:\ } }{\hfill$\Box$ }
\newtheorem {property}{Property}

\begin{document}
\setcounter{page}{1}

\renewcommand{\thesection}{\Roman{section}}


%
\title{Reachability Problem in Non-uniform Cellular Automata}

\address{maths.sumit@gmail.com}

\author{Sumit Adak\\
Department of Information Technology\\
Indian Institute of Engineering Science and Technology, Shibpur\\ Howrah-711103, India\\
maths.sumit@gmail.com
\and Sukanya Mukherjee\\
Department of Computer Science and Engineering\\ 
Institute of Engineering and Management, Kolkata\\
West Bengal 700091, India\\
sukanya.mukherjee@iemcal.com
\and Sukanta Das\\
Department of Information Technology\\
Indian Institute of Engineering Science and Technology, Shibpur\\ Howrah-711103, India\\
sukanta@it.iiests.ac.in}

\maketitle

\runninghead{S. Adak, S. Mukherjee, S. Das}{Reachability Problem in Non-uniform Cellular Automata}

\begin{abstract}
This paper deals with the CREP (Configuration REachability Problem) for non-uniform cellular automata (CAs). The cells of non-uniform CAs, we have considered here, can use different Wolfram's rules to generate their next states. We report an algorithm which decides whether or not a configuration of a given (non-uniform) cellular automaton is reachable from another configuration. A characterization tool, named {\em Reachability tree}, is used to develop theories and the decision algorithm for the CREP. Though the worst case complexity of the algorithm is exponential in time and space, but the average performance is very good.
\end{abstract}


\section{Introduction}
\label{intro}

Cellular automata (CAs) are discrete dynamical systems which produce complex global behaviour using simple local computation~\cite{Neuma66,wolfram84b}. The Configuration REachability Problem (CREP) in CAs asks to decide whether a ({\em destination}) configuration $D$ of a given cellular automaton (CA) is reachable from another ({\em source}) configuration $S$ of the CA~\cite{rpa}. The CREP is undecidable for 1-d infinite CAs~\cite{rpa1}, so researchers considered this problem for finite CAs~\cite{rpa1,rpa}. CREP is P-complete, NP-complete and PSPACE-complete depending on the types of CAs~\cite{rpa1}. It has also been shown that CREP is NP-intermediate for the CAs with additive rules~\cite{rpa}. Wolfram's rule 90, for example, is an additive rule~\cite{martin84}, and so to decide reachability of $D$ of rule 90 CA with $n$ cells from $S$, we need superpolynomial time.

However, all the works on CREP consider the classical CAs, where the cells follow same next state function (that is, $rule$) to generate their next states. In recent time, a new class of CAs, known as non-uniform CAs, are under the focus of CAs research where the cells of a CA can follow different next state functions~\cite{ppc1,Dennunzio,sipper96evnca}. Obviously, classical CAs are proper subset of these non-uniform CAs. Primary focus of the non-uniform CA research was on the one-dimensional CAs, where the cells follow Wolfram's CA rules~\cite{ppc1}. Researchers already studied the reachability problem~\cite{rpa1,rpa} for finite classical CAs. However, for non-linear non-uniform CAs, there is no method to deal with the reachability problem. In this work, we propose a method to deal with the reachability problem for 1-d finite non-uniform CAs.

We use here a characterization tool, named {\em Reachability tree}, to discover the properties of non-uniform CAs. An algorithm to decide reachability of $D$ from $S$ of a given $n$-cell non-uniform CA is reported. The algorithm can obviously deal with classical CAs as well. Worst case time complexity of the algorithm, however, is exponential, because CREP is itself PSPACE-complete~\cite{rpa}. But, the average case time requirement of the algorithm is polynomial.

To understand average case performance, we conduct an experimentation. And through experimentation, we determine that the average case complexity of the algorithm is $O(n^3)$, where $n$ is the size of automaton. 

Hereafter, by ``CA'', we will mean ``non-uniform'' CA. We next proceed with some useful definitions about CAs.

\section{Definitions}
\label{CA}
The CAs, we consider here, consist of a finite number of cells which are organized as a 1-dimensional lattice $\mathcal{L}$. The cells can be in state 0 or state 1. A configuration or (global) state of the CA is a mapping $c$: $\mathcal{L}$ $\mapsto$ $\{0,1\}$. Let us consider that $\mathcal{C}$ is the collection of all possible configurations of an $n$-cell CA (that is $|\mathcal{C}|$=$2^n$). Then, a CA is a function $F$: $\mathcal{C}$ $\rightarrow$ $\mathcal{C}$, which satisfies the following conditions: $y=F(x)$, $x,y \in {\mathcal{C}}$, where $x=(x_i)_{0\leq i\leq {n-1}}$, $y=(y_i)_{0\leq i\leq {n-1}}$ and $y_i=f_i(x_{i-1},x_i,x_{i+1})$. The $f_i:\{0,1\}^3 \mapsto \{0, 1\}$ is a next state function for the cell $i$. In this work, we consider null boundary condition where left and right neighbors of cell 0 and cell $n-1$ are always in state 0. That is, $y_0=f_0(0,x_1,x_2)$ and $y_{n-1}=f_i(x_{n-2},x_{n-1},0)$.
\begin{table}[h]
	\begin{center}	
		\caption{Rules 9, 170, 195 and 80}	
		\label{tt}
		\begin{tabular}{cccccccccc}\hline
		Present~state &  111 & 110 & 101 & 100 & 011 &  010 &  001 &  000 & Rule \\
		(RMT)& (7) & (6) & (5) & (4) & (3) & (2) & (1) & (0) &  \\\hline
		{\rm (i)~Next~state}    &   0  &  0 &  0  &  0  &   1  &   0  &   0  &   1  &  9 \\
{\rm (ii)~Next~state}    &   1 &  0  &  1  &  0  &   1  &   0  &   1  & 0 & 170 \\
{\rm (iii)~Next~state}    &   1 &  1  &  0  &  0  &   0  &   0  &   1  & 1 & 195 \\
{\rm (iv)~Next~state}    &   0 &  1  &  0  &  1  &   0  &   0  &   0  & 0 & 80 \\\hline
\end{tabular}
	\end{center}
\end{table} 

The next state function $f_i$ can be expressed in tabular form (Table~\ref{tt}). Decimal equivalents of 8-next states are conventionally called as ``rule'' ($\mathcal{R}_i$)~\cite{Wolfr83}. We name each of the 8 combinations of $x_{i-1}$, $x_i$ and $x_{i+1}$ as Rule Min Term (RMT), which is generally presented in its decimal equivalent. The 001 of the first row of Table~\ref{tt} is the RMT 1, next state against which is 0 for rule 9, 1 for rule 170. If $r$ is an RMT of ${\mathcal R_i}$, we write ${\mathcal R_i}[r]$ to denote its next state. Hence, 9[1]=0, 170[1]=1 (see Table~\ref{tt}).

Now, we introduce a set $Z_8^i$ that contains the valid RMTs of ${\mathcal R_i}$. That is, $Z_8^i = \{k ~|$ RMT $k$ of ${\mathcal R_i}$ is valid\}. Generally, $|Z_8^i|=8$. However, only four RMTs are valid for the first and last rules of a null boundary CA, and $Z^{0}_8=\{0,1,2,3\}$ and $Z^{n-1}_8=\{0,2,4,6\}$.

Traditionally, the cells of a CA follow same rule. Such a CA is {\em uniform CA}. In a {\em non-uniform CA}, the cells may follow different rules. We, therefore, need a rule vector ${\mathcal R}=\langle {\mathcal R_0},~ {\mathcal R_1}, \cdots, {\mathcal R_i}, \cdots, {\mathcal R_{n-1}}\rangle$ to define an $n$-cell non-uniform CA, where the cell $i$ follows ${\mathcal R_i}$. The uniform CA, hence, is a special case of non-uniform CA; where ${\mathcal R_0}= {\mathcal R_1}= \cdots ={\mathcal R_i}= \cdots ={\mathcal R_{n-1}}$.
\begin{figure*}[h]
	\centering
	\includegraphics[width= 2.4in, height = 1.6in]{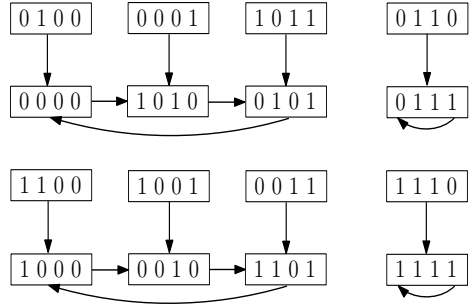}
	\caption{State transition diagram of CA $\langle 9, 170, 195, 80 \rangle$} 
	\label{Fig_STA}
\end{figure*}

{\noindent \bf State Transition Diagram:} The sequence of configurations or states of a CA generated (state transitions), during its evolution (with time), directs the CA behaviour. The state transition diagram of an automaton shows the transition of states, and depicts the relations among states of the automaton. As a proof of concept, Fig.~\ref{Fig_STA} shows the state transition diagram of a 4-cell CA $\langle$9, 170, 195, 80$\rangle$. In this work, however, we have used the terms ``configuration'' and ``state of a CA'' interchangeably.
\begin{defnn}
A state $c\in \mathcal{C}$ of a CA is {\textbf{reachable}} if there exists at least one state $x\in \mathcal{C}$ so that $c=F(x)$. If no such $x$ exists, $c$ is {\textbf{non-reachable}}.
\end{defnn}
For example, state 0011 of Fig.~\ref{Fig_STA} is non-reachable whereas state 1101 is reachable.
\begin{defnn}
A state of a CA {\textbf{$D$ is reachable from $S$}}, $D, S\in \mathcal{C}$, if there exists a finite $t\in \mathbb{N}$ so that $D=F^t(S)$. If no such $t$ exists, then \textbf{$D$ is not reachable from $S$}.
\end{defnn}
For example, state 1010 of Fig.~\ref{Fig_STA} is reachable from the state 0100. However, 0100 is not reachable from the state 1010. Please note here that ``$D$ is not reachable from $S$'' does not necessarily imply that ``$D$ is non-reachable''. $D$ may be reachable from other configuration, but not from $S$. 

{\noindent \bf RMT Sequence (RS):}
A CA state can also be viewed as a sequence of RMTs. For example, the state 0101 in null boundary condition can be viewed as $\langle 1252\rangle $, where 1, 2, 5 and 2 are the RMTs on which the transition of first, second, third and fourth cells can be made. For an $n$-bit state, we get a sequence of $n$ RMTs. Obviously, two consecutive RMTs in an RS, $r_i$ and $r_{i+1}$ are related, and $r_{i+1}$ = $2r_i$ or $2r_i+1 \pmod 8$~(Table \ref{transition}).

\begin{table}[h]
	\begin{center}	
		\caption{Relationship between $i^{th}$ and $(i+1)^{th}$ RMTs.}	
		\label{transition}
		\begin{tabular}{ccccccccc}\hline
$i^{th}$ RMT & 0 & 1 & 2 & 3 & 4 & 5 & 6 & 7 \\
$(i+1)^{th}$ RMT & 0, 1 & 2, 3 & 4, 5 & 6, 7 & 0, 1 & 2, 3 & 4, 5 & 6, 7 \\\hline
\end{tabular}
	\end{center}
\end{table} 

\begin{defnn}
Two RMTs $r$ and $s$ ($r\ne s$) are said to be equivalent to each other if $2r\equiv 2s\pmod8$.
{~\cite{SukantaTH}}
\end{defnn}

\begin{defnn}
Two RMTs $r$ and $s$ ($r\ne s$) are said to be sibling to each other if $\lfloor r/2\rfloor$=$\lfloor s/2\rfloor$.
{~\cite{SukantaTH}}
\end{defnn}
Therefore, RMT 2 is equivalent to RMT 6, whereas RMTs 2 and 3 are sibling to each other.

Now to decide whether a configuration or a state $D$ of a (non-uniform) CA reachable from another configuration $S$, we next introduce a tool, named {\em reachability tree}.

\section{Reachability Tree and Configuration Tracing}
\label{RT}

Reachability Tree~\cite{Acri04,JCA16}, a characterization tool for 1-dimensional CA, is a rooted and edge-labelled binary tree that represents the reachable states of a CA. For an $n$-cell CA, there are $n+1$ levels - root at level 0, and leaves at level $n$. We represent a node of the tree by $N_{i.j}$, where $i$ ($0\le i \le n$) is the level index, and $j$ ($0\le j \le 2^i-1$) is the node number at $i^{th}$ level. The numbering of nodes in each level starts from left side. In the reachability tree, the nodes are the subset of RMTs of rules -- $N_{i.j} \subseteq Z_8^{i}$. 

The root is formed with the RMTs of ${\mathcal R}_{0}$, the nodes of level $(n-1)$ are formed with RMTs of ${\mathcal R}_{n-1}$, and the leaf nodes are empty. We represent an edge of the tree by $E_{i.j}$, where $i$ ($0\le i \le n-1$) is the level index, and $j$ ($0\le j \le 2^{i+1}-1$) is the edge number at $i^{th}$ level. Here, we define the level of an edge. An edge is said to be edge of $i^{th}$ level, if it connects the nodes of $i^{th}$ and $(i+1)^{th}$ levels. So, we can write, $E_{i.2j} = (N_{i.j}, N_{i+1.2j}, l_{i.2j})$ and $E_{i.2j+1} = (N_{i.j}, N_{i+1.2j+1}, l_{i.2j+1})$ ($0\le i\le n-1$, $0\le j\le 2^i-1$), where $l_{i.2j}\subseteq N_{i.j}$ and $l_{i.2j+1}\subseteq N_{i.j}$ are the labels of the edges, and $l_{i.2j}\cup l_{i.2j+1} = N_{i.j}$. If $l_{i.k}=\emptyset$ for any $k$, the edge $E_{i.k}$ (hence, $N_{i+1.k}$) does not exit. We call such an edge as a non-reachable edge. However, for each $r\in l_{i.2j}$ (resp. $r\in l_{i.2j+1}$), RMT $r$ of ${\mathcal R}_{i}$ is 0 (resp. 1) and we get two RMTs $2r\pmod8$ and $2r + 1\pmod{8}$ of ${\mathcal R}_{i+1}$ in $N_{i+1.2j}$ (resp. $N_{i+1.2j+1}$), and the edge is called 0-edge (resp. 1-edge). Following is the formal definition of the reachability tree.

\begin{defnn}
\label{Def:RTNull}
Reachability tree of an $n$-cell CA with rule vector $\langle$${\mathcal R}_0$, ${\mathcal R}_1$, $\cdots$ , ${\mathcal R}_i$, $\cdots$, ${\mathcal R}_{n-1}$$\rangle$ under null boundary condition is a rooted and edge-labelled binary tree with $n+1$ levels, where $E_{i.2j} = (N_{i.j}, N_{i+1.2j}, l_{i.2j})$ and $E_{i.2j+1} = (N_{i.j}, N_{i+1.2j+1}, l_{i.2j+1})$ are the edges between nodes $N_{i.j}\subseteq Z_8^{i}$ and $N_{i+1.2j}\subseteq Z_8^{i+1}$ with label $l_{i.2j}\subseteq N_{i.j}$, and between nodes $N_{i.j}$ and $N_{i+1.2j+1}\subseteq Z_8^{i+1}$ with label $l_{i.2j+1}\subseteq N_{i.j}$ respectively $(0\le i\le n-1$, $0\le j\le 2^i-1)$. Following are the relations which exist in the tree:

\begin{enumerate}

\item \label{rootAtDefNull} {[For root]} $N_{0.0} = Z_8^0 = \{0, 1, 2, 3\}$.

\item \label{edgeAtDefNull} $\forall r\in N_{i.j}$, RMT $r$ of ${\mathcal R}_{i}$ is in $l_{i.2j}$ (resp. $l_{i.2j+1}$), if ${\mathcal R}_{i}[r]$ = 0 (resp. 1). That means, $l_{i.2j}\cup l_{i.2j+1} = N_{i.j}$ ($0\le i\le n-1$, $0\le j\le 2^i-1 $).

\item \label{nodeAtNodeNull} $\forall r \in l_{i.j}$, RMTs $2r \pmod{8}$ and $2r+1 \pmod{8}$ of ${\mathcal R}_{i+1}$ are in $N_{i+1.j}$ ($0\le i\le n-3$, $0\le j\le 2^{i+1}-1$).

\item \label{n-1AtNodeNull} {[For level $n-1$]} $\forall r \in l_{n-2.j}$, RMT $2r \pmod{8}$ of ${\mathcal R}_{n-1}$ is in $N_{i+1.j}$ ($0\le j\le 2^{n-1}-1$).

\item \label{nAtNodeNull} {[For level $n$]} $N_{n.j}=\emptyset$, for any $j$, $0\le j\le 2^n-1$.

\end{enumerate}
\end{defnn}


\begin{figure*}[h]
	\centering
	\includegraphics[width= 5.0in, height = 3.2in]{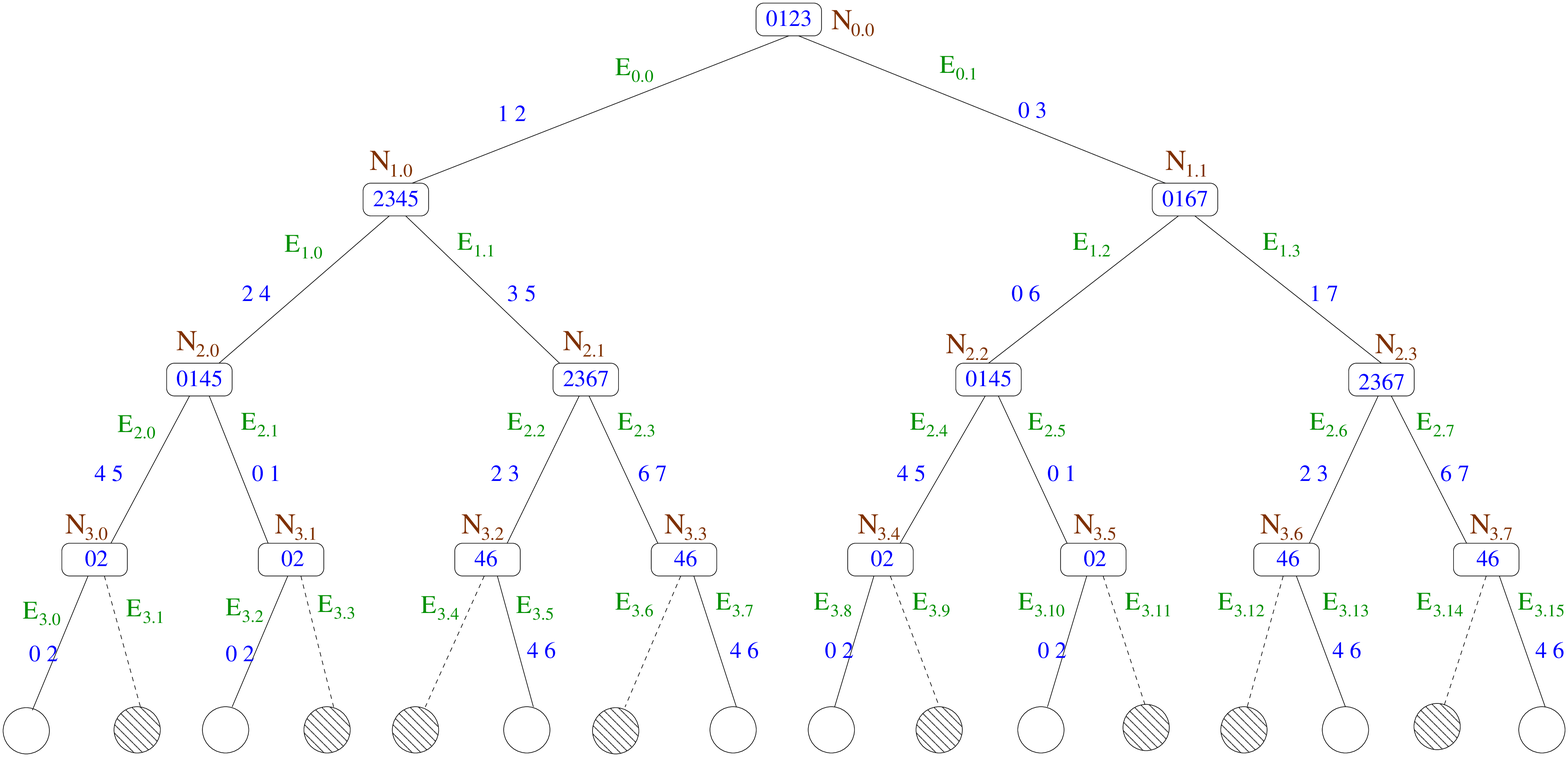}
	\caption{Reachability tree of CA $\langle 9, 170, 195, 80 \rangle$} 
	\label{rtree}
\end{figure*}
{Fig.~\ref{rtree}} is the reachability tree of the CA of Fig.~\ref{Fig_STA}. According to the null boundary condition, only 4 RMTs (0, 1, 2 and 3) of ${\mathcal{R}}_0$ are valid, and so the root is formed with these 4 RMTs. That is, $N_{0.0}=Z_8^0=\{0, 1, 2, 3\}$. Similarly, $Z_8^{n-1}=Z_8^3=\{0, 2, 4, 6\}$ and $N_{3.j}\subseteq Z_8^3$ for all $j$, $0\leq j\leq 3$. However, the label of edge $E_{0.1}$ is \{0, 3\}, as RMTs 0 and 3 of rule 9 are 1. We write RMTs of a label on the edge. Note that, the label of $E_{3.1}$ is empty, that is, $l_{3.1}=\emptyset$. This edge is {\em non-reachable}, and it can not connect any node of next level. Fig.~\ref{rtree} uses dotted line for them. Since $Z_8^{n}=\emptyset$ for an $n$-cell CA, the leaves are empty. The number of leaves (excluding dotted leaves) in Fig.~\ref{rtree} is 8, which is the number of reachable states. We call edge $E_{i.j}$ as {\em 0-edge} when $j$ is even, and {\em 1-edge} otherwise. We further call the edge $E_{i.j}$ as an edge of level $i$. A sequence of edges from the root to a leaf node represents a reachable state, when 0-edge and 1-edge are replaced by 0 and 1 respectively. For example, 0000 is a reachable state in {Fig.~\ref{rtree}}, but the state 0001 is non-reachable.

From the reachability tree, we can get the information about reachable and non-reachable states. A sequence of edges $\langle$$E_{0.j_0}$ $E_{1.j_1}$ $\cdots$ $E_{i.j_i}$ $E_{i+1.j_{i+1}}$ $\cdots$ $E_{n-1.j_{n-1}}$$\rangle$ from root to a leaf associates a reachable state and at least one RS $\langle r_0r_1\cdots r_ir_{i+1}\cdots r_{n-1}\rangle$, where $r_i\in l_{i.j_i}$ and $r_{i+1}\in l_{i+1.j_{i+1}}$ ($0\le i<n-1$, $0\le j_i\le 2^i-1$, and $j_{i+1}=2j_i$ or $2j_i+1$). That is, the sequence of edges represents at least two CA states. Note that if RMT $r_i$ is 0 (resp. 1) then $E_{i.j_i}$ is 0-edge (resp. 1-edge). Therefore, the reachable state is the next (resp. present) state of the state (resp. predecessor), represented as RS. Interestingly, there are $2^n$ RSs in the tree, but number of reachable states may be less than $2^n$. A sequence of edges may associate $m$-number of RSs ($m\ge 1$), which implies, this state is reachable from $m$-number of different states.

Obtaining only reachable or non-reachable states using reachability tree is not enough to make the decision about reachability of one state from another. We need to find out the predecessor(s) of each state in the reachability tree. Then only we can trace in the tree if a CA state $D$ is reachable from another state $S$. However, the tree guides us to find the predecessors of the CA states by establishing relation among edges. To find the relations among the edges, we introduce the concept of ``link'' in the next section. 

\subsection{Links}
\label{link}
As we have discussed before, a CA state/configuration can be represented as a bit sequence, and as an RMT sequence. Reachability tree uses both the representations - bit sequences to represent the reachable states, and RMT sequences to represent their predecessors. Now, the predecessors, which are also CA states, can be observed in the tree as bit sequence. Intuitively, the ``links'' link the states represented as bit sequences to their predecessor.

The links are formed for each RMT $r\in l_{i.j}$, present on edge $E_{i.j}$ ($0\le i\le n-1$, $0\le j\le 2^{i+1}-1$). By the processing of reachability tree, we find the links among the edges for each individual RMT on the tree. The links are formed depending on whether the RMTs are self replicating (defined below) or not.

\begin{sidewaysfigure}
	\centering
    \includegraphics[scale=.45]{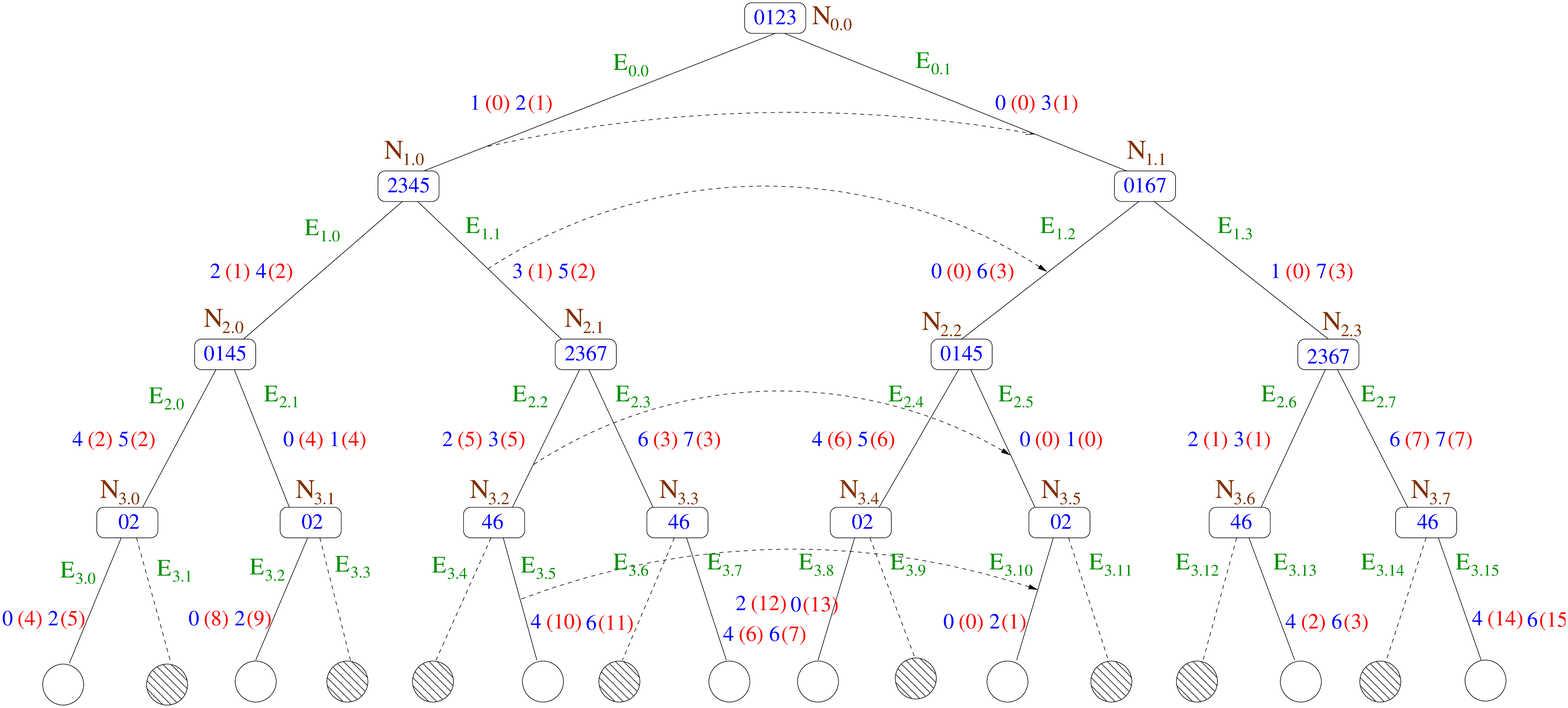}
   \caption{Links among edges of reachability tree for CA $\langle 9, 170, 195, 80 \rangle$} 
    \label{Fig_LT}
\end{sidewaysfigure}

%

\begin{defnn}
\label{stbst}
An RMT $r=4x+2y+z$ of a rule ${\mathcal R_i}$ is said to be self replicating if ${\mathcal R_i}[r]=y$ where $x,y,z\in\{0,1\}$.
\end{defnn}

For example, RMT 1 (001) and RMT 3 (011) of rule 9 is self replicating, whereas RMTs 4, 5, 6 and 7 of rule 195 are self replicating (see Table \ref{tt}). If an RMT $r\in l_{i.j}$ is not self replicating, then there is a link from the edge $E_{i.j}$ to $E_{i.k}$ ($j\ne k$). Depending on the values of $j$ and $k$, we can classify the links in the following way: forward link (when $j<k$), backward link (when $j>k$) and self link (when $j=k$). We represent this link as $E_{i.j}(r)\longrightarrow E_{i.k}$. The rules, followed to form links in a reachability tree, are noted below:\\

\noindent {\bf R1)} If RMT $r\in l_{0.j}$ is self replicating ($j= 0$ or $1$), the edge $E_{0.j}$ is self linked for RMT $r$. Otherwise, if $j=0$, there is a forward link from $E_{0.0}$ to $E_{0.1}$ for RMT $r$; else, there is a backward link from $E_{0.1}$ to $E_{0.0}$ for RMT $r$.\\

\noindent {\bf R2)} If $E_{i-1.j}$ is self linked for RMT $r\in l_{i-1.j}$, and if $s$ is self replicating where $s\in l_{i.2j}$ (resp. $s\in l_{i.2j+1}$) is $2r$ or $2r+1\pmod8$, then $E_{i.2j}$ (resp. $E_{i.2j+1}$) is self linked. But if $s$ is not self replicating, then there is a forward link from $E_{i.2j}$ to $E_{i.2j+1}$ (resp. backward link from $E_{i.2j+1}$ to 
$E_{i.2j}$).\\

\noindent {\bf R3)} If there is a link from $E_{i-1.j}$ to $E_{i-1.k}$ ($j\ne k$) for RMT $r\in l_{i-1.j}$, and $s\in l_{i.2j}$ (resp. $s\in l_{i.2j+1}$) is $2r$ or $2r+1\pmod8$, then there is a link from $E_{i.2j}$ (resp. $E_{i.2j+1}$) to $E_{i.2k}$ while $s\in \{0, 1, 4, 5\}$ or to $E_{i.2k+1}$ while $s\in \{2, 3, 6, 7\}$. It is forward link if $j<k$, backward link if $j>k$.

\begin{example}
Fig.~\ref{Fig_LT} shows the links of edges caused by RMTs of the CA $\langle$9, 170, 195, 80$\rangle$. There is a (forward) link from $E_{0.0}$ to $E_{0.1}$ for RMT 2, so we write the link within a bracket beside the RMT 2. Now, we get a forward link from $E_{1.1}$ to $E_{1.2}$ for RMT 5. Now, we get $E_{2.2}(2)$ $\rightarrow$ $E_{2.5}$, and $E_{3.5}(4)$ $\rightarrow$ $E_{3.10}$. Therefore, for the RS $\langle 2524\rangle$, we can get a sequence of links, hence a sequence of edges $\langle E_{0.1}E_{1.2}E_{2.5}E_{3.10}\rangle$, which represents 1010. Note that the RS $\langle 2524\rangle$ corresponds to the state 1010. The sequence $\langle E_{0.0}E_{1.1}E_{2.2}E_{3.5}\rangle$ associates the state 0101, as well as the RS $\langle 2524\rangle$. The RS $\langle 2524\rangle$, hence the state 1010, is the predecessor of the state 0101. See Fig.~\ref{Fig_STA} for verification.
\end{example}

The links help us to trace state transitions in reachability tree by identifying the predecessor(s) of each state. Through the links, we can identify the predecessor of predecessor of a state. If $E_{i.j}$ is linked with $E_{i.k}$ for RMT $r_1\in l_{i.j}$, and $E_{i.k}$ is linked with $E_{i.p}$ ($0\le j < k < p\le 2^i-1$ for forward link, $2^i-1\geq j > k > p \geq 0$ for backward link) for RMT $r_2\in l_{i.k}$, we say
there exists a link (forward or backward) from $E_{i.j}$ to $E_{i.p}$, where $1 \le i\le n-1$. Therefore, we get the following property (transitivity property) of the links. We write $E_{i.j}(r_1 )\rightarrow E_{i.k}$, if there is a link from $E_{i.j}$ to $E_{i.k}$ for RMT $r_1\in l_{i.j}$.
\begin{itemize}
\item If $E_{i.j}(r_1)\rightarrow E_{i.k}$ and $E_{i.k}(r_2)\rightarrow E_{i.p}$, then
\item $E_{i.j}(r_1)\rightarrow E_{i.k}(r_2)\rightarrow E_{i.p}$.
\end{itemize}
Now, we define {\em length} of the links. If from edge $E_{i.j_1}$ to $E_{i.j_2}$, there are $k$ number of RMTs (or $k$ number of edges), then we write: $length(E_{i.j_1}, E_{i.j_2})=k$. We write, $length(E_{i.j_1}, E_{i.j_2}) = \infty$ if there is no link between $E_{i.j_1}$ and $E_{i.j_2}$. In Fig.~\ref{Fig_LT}, following connection between $E_{1.0}$ and $E_{1.3}$ exists: $E_{1.0}(4) \rightarrow E_{1.2}(6) \rightarrow E_{1.3}$. That is, $length(E_{1.0}, E_{1.3})=2$.

\begin{lemma}
\label{lem1}
There exist only two links to $E_{i.j}$ from any one or two edges for RMTs $r$ and $s$ when $0 \le i < n-1$ and $r$ and $s$ are sibling to each other, and only one link when $i=n-1$ in a reachability tree ($0\le j\le 2^{i+1}-1$).~\cite{JCA16}
\end{lemma}

\begin{property}
\label{pr1}
A link present at $i^{th}$ level triggers two links at level $i+1$, where $0\le i\le n-3$, a link of $(n-2)^{th}$ level derives one link at $(n-1)^{th}$ level.
\end{property}

This is obvious, because an RMT $r$ at a node/label of level $i$ contributes two RMTs - $2r\pmod 8$ and $2r+1\pmod 8$ in node/label(s) of level $i+1$. Both the RMTs participate in links, depending upon the link caused by RMT $r$. For example, the link $E_{0.1}(0)\rightarrow E_{0.0}$ triggers two links $E_{1.2}(0)\rightarrow E_{1.0}$ and $E_{1.3}(1)\rightarrow E_{1.0}$. However, a link at level $n-2$ triggers only one link at last level, as RMT $2r+1 \pmod 8$ is invalid in that level.

Let us now define {\em path} between two edges of a level - $E_{i.j_1}$ and $E_{i.j_k}$. We say that there exists a path between $E_{i.j_1}$ and $E_{i.j_k}$ if $E_{i.j_1}$ is linked to $E_{i.j_k}$, that is, if $length(E_{i.j_1},E_{i.j_k})$ is finite. Otherwise, there is no path between $E_{i.j_1}$ and $E_{i.j_k}$. If a path exists, we write it as the following: $E_{i.j_1}(r_1)$ $\rightarrow$ $E_{i.j_2}(r_2)$ $\rightarrow$ $\cdots$ $\rightarrow$ $E_{i.j_k}$. Now, the question is, can we say that there exist a path between $E_{i+1.p}$ and $E_{i+1.q}$ where $p \in \{2j_1, 2j_1+1\}$ and $q \in \{2j_k, 2j_k+1\}$? No, not always. Following Property~\ref{pr1}, if a path is formed from $E_{i+1.p}$ and $E_{i+1.q}$ due to the path between $E_{i.j_1}$ and $E_{i.j_k}$, we say the path between $E_{i+1.p}$ and $E_{i+1.q}$ is triggered by the path between $E_{i.j_1}$ and $E_{i.j_k}$. However, no path may be triggered at level $i+1$. Obviously, a path from $E_{n-1.j_1}$ to $E_{n-1.j_k}$ is triggered by the paths above.

\begin{example}
In Fig.~\ref{Fig_LT}, following path is formed at level $0$, which and triggers a path at level $3$: $E_{0.1}(0) \rightarrow E_{0.0}$, $E_{1.2}(0) \rightarrow E_{1.0}$, $E_{2.5}(0) \rightarrow E_{2.0}$, $E_{3.10}(0) \rightarrow E_{3.0}$.
\end{example}

Now, we explore the reachability tree to check that if there exists any path or not from destination edge ($D$) to source edge ($S$) at leaf level.
\section{Reachability Analysis}
\label{RoN}

To check whether a configuration $D$ of an $n$-cell CA is reachable from another configuration $S$, we rewrite the configurations as following: $S=(s_i)_{0\leq i\leq n-1}$ and $D=(d_i)_{0\leq i\leq n-1}$. The configurations can also be identified in the reachability tree as sequences of edges. For ease of understanding, let us rename the sequences of edges as $(s_i)_{0\leq i\leq n-1}$ representing $S$, and as $(d_i)_{0\leq i\leq n-1}$ representing D. Now, we search in the reachability tree for a path from $d_i$ to $s_i$. If no path exists, we declare that $D$ is not reachable from $S$.

\begin{theorem}
\label{th1}
For an $n$-cell CA, $D$ is reachable from $S$, if and only if there exists a path from $d_{n-1}$ to $s_{n-1}$.
\end{theorem} 

\begin{proof}
Let us consider, there is a path from $d_{n-1}$ to $s_{n-1}$ at leaf level of length $m$: $E_{n-1.j_1}(r_1)$ $\rightarrow$ $\cdots$ $\rightarrow$ $E_{n-1.j_q}(r_q)$ $\rightarrow$ $\cdots$ $\rightarrow$ $E_{n-1.j_m}$ where $d_{n-1}$ = $E_{n-1.j_1}$ and $s_{n-1}$ = $E_{n-1.j_m}$. Now, we can proof $D$ is reachable from $S$. Hence, we can get a sequence of edges from root to $E_{n-1.k}$ for each $k\in \{j_1, j_2, \cdots , j_m\}$ which represents a reachable state. Here, two reachable states which are represented by edge sequences that end with $E_{n-1.j_p}$ and $E_{n-1.j_{p+1}}$ respectively are two consecutive states. Hence, we can get a sequence of consecutive states. Since there is a path, the sequence of states forms a path involving the RMTs $r_1, r_2,\cdots, r_{m-1}$. Hence, $D$ is reachable from $S$.

Now suppose, $D$ is reachable from $S$. Obviously, there is a path from $d_i$ to $s_i$, $0\leq i\leq n$. Hence the proof.
\end{proof} 
 
\begin{example}
Suppose, $S=0000$ and $D=0101$ for the CA $\langle 9, 170, 195, 80 \rangle$. Now, from Fig.~\ref{Fig_LT}, we see that $d_3$ = $E_{3.5}$ and $s_3$ = $E_{3.0}$. From the linked tree, we can get the path - $E_{3.5}(4) \rightarrow E_{3.10}(0)\rightarrow E_{3.0}$ ($length(d_3,s_3) = 2$). Therefore, $D$ is reachable from $S$, and $D=F^2(S)$ (check it from Fig.~\ref{Fig_STA}).

For the same CA, if $S=0000$ and $D=1101$, then $d_3$ = $E_{3.13}$ and $s_3$ = $E_{3.0}$. From Fig.~\ref{Fig_LT}, we can see that there is no path from $d_3$ to $s_3$. So, $D$ is not reachable from $S$.
\end{example}  

To decide the reachability, we first form the root of the reachability tree (using $\mathcal{R}_0$), get edges from the root, identify links between edges following rule \textbf{R1} of link formation. Then, check if there exist any path from $d_0$ to $s_0$. If it exists then we continue, otherwise conclude that $D$ is not reachable from $S$. If it exists then form the next level (using $\mathcal{R}_1$) and get the links, and again check whether there exists any path from $d_1$ to $s_1$. If no path exists, then $D$ is not reachable from $S$. Otherwise, continue the same process. Finally, if there exist a path from $d_{n-1}$ to $s_{n-1}$, then declare that $D$ is reachable from $S$.

By definition, reachability tree grows exponentially, in general. In this particular problem, however, we do not deal with all the edges. The edges, not in the path of $d_i$ and $s_i$, are irrelevant to us. To reduce the number of edges/nodes in the proposed decision procedure, we remove such irrelevant edges.

\begin{figure*}
	\centering
	\includegraphics[width= 4.2in, height = 2.4in]{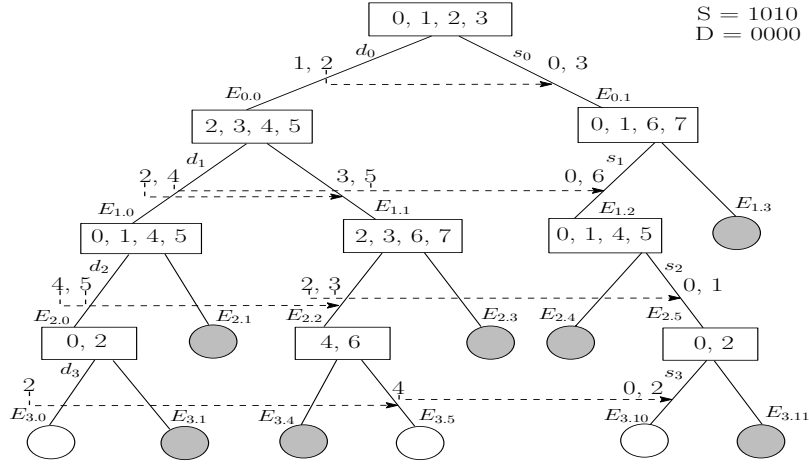}
	\caption{Unnecessary edges for a particular $S$ and $D$ can be removed}
	\label{RT_cut}
\end{figure*}
 
\begin{example}
Let us consider the CA $\langle 9, 170, 195, 80 \rangle$ and $S=1010$ and $D=0000$. Fig.~\ref{RT_cut} explains that $D$ is reachable from $S$. The paths of $d_i$ and $s_i$ are shown in the figure. The edge $E_{1.3}$ is not in the path of $d_1$ and $s_1$, and so it is irrelevant in this particular case. Hence, $E_{1.3}$ is removed, and the corresponding sub tree is not further developed. Similarly, $E_{2.1}$, $E_{2.3}$, $E_{2.4}$, $E_{3.1}$, etc are irrelevant, and hence removed. Obviously, we need not to deal with a good number of edges/nodes here.
\end{example} 

However, we can sometime decide the non-reachability of $D$ from $S$ without tracing path of $d_i$ and $s_i$, but by observing some conditions related to $d_i$ and $s_i$. We next report these conditions.
 
\begin{condition}
\label{con1}
For an $n$-cell CA, if the edge $d_i$ is non-reachable where $0\leq i\leq n-1$, then $S$ to $D$ is not reachable. 
\end{condition}

\begin{reason}
From Theorem~\ref{th1}, we know that, $D$ is reachable from $S$ if there exist a path from $d_{n-1}$ to $s_{n-1}$. And from Property~\ref{pr1}, we can say that the path at leaf level is triggered from the root. That is, if there is a path from $d_{n-1}$ to $s_{n-1}$, then there are the paths from $d_0$ to $s_0$, $d_1$ to $s_1$, $\cdots$ , and $d_{n-2}$ to $s_{n-2}$. If at any level, $d_i$ is non-reachable, then there is no link from this edge. Hence, there is no path from $d_{i+1}$ to $s_{i+1}$, $\cdots$ , and $d_{n-1}$ to $s_{n-1}$. Therefore, there is no path at leaf level and we can conclude that $S$ to $D$ is non-reachable.
\end{reason}

\begin{example}
Suppose, $S=0000$ and $D=1011$ for the CA $\langle 9, 170, 195, 80 \rangle$. Now, from Fig.~\ref{Fig_LT}, we get $d_0$ = $E_{0.1}$ and $s_0$ = $E_{0.0}$, and there is a path $E_{0.1}(0)\rightarrow E_{0.0}$. At the second level, $d_1$ = $E_{1.2}$ and $s_1$ = $E_{1.0}$ and there is also a path: $E_{1.2}(0)\rightarrow E_{1.0}$. Now, at the third level, $d_2$ = $E_{2.5}$ and $s_2$ = $E_{2.0}$ and there is also a path: $E_{2.5}(0)\rightarrow E_{2.0}$. Now at the leaf level, $d_3$ = $E_{3.11}$ and $s_3$ = $E_{3.0}$, but the edge $d_3$ is non-reachable edge. So, there is no path from $d_3$ to anywhere. Therefore, $D$ is not reachable from $S$ (see Fig.~\ref{Fig_STA}).
\end{example}

\begin{condition}
\label{con2}
For an $n$-cell CA, if the edge $s_i$ is self linked for two sibling RMTs and $d_i\neq s_i$, then $S$ to $D$ is not reachable ($0\leq i\leq n-1$).
\end{condition}

\begin{reason}
From Theorem~\ref{th1}, we know that, $D$ is reachable from $S$ if there exists a path from $d_{n-1}$ to $s_{n-1}$, which immediately implies the paths from $d_0$ to $s_0$, $d_1$ to $s_1$, $\cdots$ , and $d_{n-2}$ to $s_{n-2}$. From Lemma~\ref{lem1}, we get that there exist two links to $E_{i.j}$ from any edges (except leaf level). If the edge $s_i$ is self linked for two RMTs, then no other edge can link to $s_i$. So, we can reach to $s_i$ from only the edge $s_i$ and if $d_i\neq s_i$, there is no path from $d_i$ to $s_i$.
\end{reason}
\begin{example}
Consider, $S=1111$ and $D=0000$ of the CA $\langle 9, 170, 195, 80 \rangle$. Now, from Fig.~\ref{Fig_LT}, we get that $d_0$ = $E_{0.0}$ and $s_0$ = $E_{0.1}$ and there is path $E_{0.0}(2)\rightarrow E_{1.0}$. Now, at the next level, $d_1$ = $E_{1.0}$ and $s_1$ = $E_{1.3}$ and there is also have path: $E_{1.0}(4)\rightarrow E_{1.2}(6)\rightarrow E_{1.3}$. At the next level, $d_2$ = $E_{2.0}$ and $s_2$ = $E_{2.7}$ and there is no path from $d_2$ to $s_2$ ($d_i\neq s_i$). The edge $s_2$ is self linked for RMTs 6 and 7 of rule 195. Hence, the edge is not reachable from any other edge.
\end{example} 

\section{Decision Algorithm}
\label{Alg}
Now, we present an algorithm to decide whether $S$ to $D$ is reachable or not. The following algorithm uses the theories framed in the earlier sections, to decide the same. However, the algorithm deals only with the labels of edges. Moreover, the algorithm does not form the whole tree at a time, but it deals with two sets of labels - \{$l_{i.0}, l_{i.1}, \cdots l_{i.2^i-1}$\} and \{$l_{i+1.0}, l_{i+1.1}, \cdots l_{i+1.2^{i+1}-1}$\}. We proceed with only non-empty labels, $ l_{0},~l_{1},\cdots $ and $ l'_{0},~l'_{1},\cdots $. Here, $ l_j$ corresponds to the label of $ E_{i.j} $ and $ l'_{k} $ correspond to the label of $ E_{i+1.k}s$ ($ 0 \le i \le n-1 $). The input of the algorithm is the CA ({\em rule vector}), $S$ ({\em Source}) and $D$ ({\em Destination}). The output is `Yes' if $D$ is reachable from $S$; `No' otherwise.

%
%

\begin{algorithm2e}[h]
\caption{Decide reachability of $D$ from $S$}
\KwIn{\textit{ $\langle \mathcal{R}_0, \mathcal{R}_1, \cdots , \mathcal{R}_{n-1}\rangle$}, $S$ = $(s_i)_{0\leq i\leq n-1}$ and $D$ = $(d_i)_{0\leq i\leq n-1}$}
\KwOut{\textit {Yes/No}}
{\em \textbf{Step 1:}} (a) Put each valid RMT $r$ of ${\mathcal R}_0$ in $l'_0$ (resp. $l'_1$) if ${\mathcal R}_0[r]=0$ (resp. 1), and get links for each RMT. \\
\hspace{3.75em}(b) Set $s\gets s_0$ and $d\gets d_0$\\
{\em \textbf{Step 2:}} Set $Count\gets 1$, $i\gets 0$ and goto Step 6\\
{\em \textbf{Step 3:}} If $i\geq n$, report ``Yes'', and exit.\\
{\em \textbf{Step 4:}} For label $l_k$, $0\leq k \leq Count-1$\\
\hspace{3.80em}Find $l'_{2k}$ and $l'_{2k+1}$ so that, if $r\in l_k$ and $s=2r\pmod 8$ or $2r+1\pmod 8$, then $s\in l'_{2k}$ (resp. $l'_{2k+1}$) when ${\mathcal R}_i[s]=0$ (resp. ${\mathcal R}_i[s]=1$) and get the links for each RMT.\\
{\em \textbf{Step 5:}} Set $s\gets 2*s+s_i$ and $d\gets 2*d+d_i$\\
{\em \textbf{Step 6:}} Verify the following.\\
			\hspace{3.5em} (a) If $l'_d=\emptyset$, report ``No'' and exit (Condition~\ref{con1}).\\
			\hspace{3.5em} (b) If $l'_s$ is self linked for two RMTs, report ``No'' and exit (Condition~\ref{con2}).\\
{\em \textbf{Step 7:}} Search for paths from $l'_d$ to $l'_s$. If no path exists, report ``No'' and exit.\\
{\em \textbf{Step 8:}} (a) Mark the labels ($l'_j$s) which are not in any path, computed in Step 7, as {\em irrelevant}.\\
\hspace{3.75em}(b) $Count\gets 2*Count$ $-$ $\#$ $irrelevant$ $labels$.\\
{\em \textbf{Step 9:}} (a) Assign the elements of $l'$ to $l$ without irrelevant labels, and accordingly update $s$ and $d$.\\
\hspace{3.75em}(b) Set $i\gets i+1$, and goto Step 3.\\
\label{alg1}
\end{algorithm2e}

\begin{example}
Let us consider the CA $\langle$9, 170, 195, 80$\rangle$, $S=1010$ and $D=0000$ (Fig.\ref{RT_cut}) as input to Algorithm~\ref{alg1}. Here $l'_0=\{1,2\}$, $l'_1=\{0,3\}$, $s=1$ and $d=0$. A path from $l'_0$ to $l'_1$ exists (Step 7). Since there is no irrelevant label, so $Count=2$. Next, we get 4 labels (Fig.~\ref{Fig_LT}) $l'_0=\{2,4\}$, $l'_1=\{3,5\}$, $l'_2=\{0,6\}$ and $l'_3=\{1,7\}$ (Step 4). Now, $s=2$ and $d=0$. The conditions of Step 6 are not satisfied, so the algorithm searches for a path from $l'_0$ to $l'_2$. There exists a path involving $l'_0$, $l'_1$ and $l'_2$ (see Fig.\ref{RT_cut}). Obviously $l'_3$ is irrelevant in this case. Hence, $Count=3$ (Step 8(b)). Now, we assign the following: $l_0\gets l'_0$, $l_1\gets l'_1$, $l_2\gets l'_2$, and further we update $s=2$ and $d=0$ (Step 9(a)). As a next step, the algorithm finds $l'_0$, $l'_1$, $\cdots$ , $l'_5$ (Step 4) and sets $s=5$ and $d=0$. There exists a path involving $l'_0$, $l'_2$ and $l'_5$. So, $l'_1$, $l'_3$ and $l'_4$ are irrelevant in this case. Hence, $Count=3$ (Step 8(b)). Now, we assign following: $l_0\gets l'_0$, $l_1\gets l'_2$, $l_2\gets l'_5$, and further we update $s=2$ and $d=0$ (Step 9(a)). In this way, the algorithm proceeds, and finally reports ``Yes''.
\end{example}

\noindent \textbf{Correctness of Algorithm~\ref{alg1}:}  The correctness of the algorithm is directly connected to the theorems, lemmas and conditions reported before. The algorithm conceptually forms reachability tree for the given CA and finds the links at each level. From the root to leaf, at any level, if the destination edge is non-reachable or the source edge is self linked for two sibling RMTs, then according to Condition~\ref{con1} or Condition~\ref{con2}, the algorithm terminates with output {\em Non-reachable}. At any level, if there does not exist any path, then according to Property~\ref{pr1} and Theorem~\ref{th1}, the algorithm terminates with output {\em Non-reachable}. Otherwise, it forms a new level and checks the paths. At leaf level, if there exists any path, then according to Theorem~\ref{th1}, the algorithm terminates with output {\em Reachable}. 

\begin{theorem}
\label{th_alg}
The upper bound running time of Algorithm~\ref{alg1} is proportional to the number of edges explored by the algorithm. 
\end{theorem}

\begin{proof}
Algorithm~\ref{alg1} contains main loop enclosing Steps 4-9. Hence, the time complexity of the algorithm is dependent on the time requirements of the steps. However, Step 4 finds the labels of edges of a level, and Steps 5-9 work on those labels. That is, if $k$ number of labels, hence edges, are explored at Step 4, then the other labels work only with them. Therefore, the upper bound of the time requirement for single execution of Steps 4 to 9 is proportional to $k$. Now, before halting of the algorithm, it repeatedly explores the edges in each run of the main loop. Hence, upper bound of the running time is proportional to the total number of edges explored by the algorithm.
\end{proof}

\noindent{\textbf{Worst case analysis:}} The worst case in Algorithm~\ref{alg1} occurs if $D$ is reachable from $S$ and no labels (hence, edges) can be removed. That is, the reachability tree contains all the possible leaves. In that case, space requirement, which is determined by two arrays - $l_i$ and $l'_{i}$, is exponential. The time requirement is then obviously exponential.

However, the algorithm performs well on an average. Because, in many cases, many edges are removed, and before reaching to the leaf of the tree, non-reachability can be decided. A sample result of another experimentation is shown Table~\ref{tableexp}, which speaks about the fact that in many cases, we need not to deal with all the of a CA. The first rule vector of Table~\ref{tableexp} says that if $S=10(0+1)^{n-2}$ and $D=11(0+1)^{n-2}$, and if first two rules of the CA are 8 and 58, then $D$ is not reachable from $S$ for any value of $n\geq 2$. Table~\ref{tableexp} gives us an idea that reachability can be decided much before than encountering the last rule. To understand the average performance of the algorithm, we have arranged a detailed experimental study which is reported in the next section.

%
\begin{table}[h]
\setlength{\tabcolsep}{0.2pt}
	\begin{center}
	\caption{An experimental study}	
		\label{tableexp}
		
		\begin{tabular}{|c|c|c|c|c|c|c|}\hline
		 Rule Vector & CA Size & Source & Destination & Reachable & Decision & Remarks \\ 
		 & & $S$ & $D$ &  or Not &  Level & \\ \hline
		 $\langle$8, 58, $\mathcal{R}_2$, $\mathcal{R}_3$, & $n$ $\geq$ $2$ & $10$ & $11$ & Not & 1  & satisfies \\ 
		 $\cdots$ , $\mathcal{R}_{n-1}$$\rangle$ &  & $(0+1)^{n-2}$ & $(0+1)^{n-2}$ & reachable &   & Condition~\ref{con1}\\ \hline
		 
		 $\langle$10, 164, $\mathcal{R}_2$, $\mathcal{R}_3$, & $n$ $\geq$ $2$ & $00$ & $10$ & Not & 1  & satisfies \\ 
		 $\cdots$ , $\mathcal{R}_{n-1}$$\rangle$ &  & $(0+1)^{n-2}$ & $(0+1)^{n-2}$ & reachable &  & Condition~\ref{con2} \\ \hline
		 
		 
		 
		 $\langle$7, 72, 254, $\mathcal{R}_3$, & $n$ $\geq$ $3$ & $111$ & $011$ & Not & 2  & no path \\
		 $\mathcal{R}_4$, $\cdots$ , $\mathcal{R}_{n-1}$$\rangle$ &  & $(0+1)^{n-3}$ & $(0+1)^{n-3}$ & reachable &   & exists \\ \hline
		 
		 $\langle$15, 213, 5, 196, 124,& $n$ = $10$ & 01001 & 11001 & Reachable & 9  & satisfies \\ 
		 243, 218, 99, 184, $85$$\rangle$ & & 00101 & 11011 &  &  & Theorem~\ref{th1} \\ \hline
\end{tabular}
\end{center}
\end{table}

\section{Average Case Analysis}
\label{Ex-study}
We find the upper bound of average running time of Algorithm~\ref{alg1} experimentally. Theorem~\ref{th_alg} points out the fact that the running time of the algorithm is proportional to the number of edges explored in corresponding reachability tree. By the proposed experimentation, we, therefore, find the average number of edges explored by Algorithm~\ref{alg1} for a given CA size. We next proceed with experimental setup.

\subsection{Experimental Setup}
\label{Ex-setup}


In this experiment, we use simple random sampling {\em with replacement} to calculate the population mean ($\mu$)~\cite{willam,SETHI2016}. In the estimation process, ${\overline{X_k}}$ denotes the mean of $k^{th}$ sample, and ${\widehat{\overline{X_k}}}$ denotes the $k^{th}$ estimate to the population mean ($k\ge 1$). Let us consider that the sample size is $m$. So, ${\overline{X_k}}=\frac{1}{m}\sum_{i=1}^{m}{x_i}$, where $x_i$ is an element of the population which is chosen randomly and uniformly.

In the experiment, we first find ${\overline{X_1}}$ which is considered as the first estimate ${\widehat{\overline{X_1}}}$ to population mean ($\mu$). Next we take the second sample of size $m$, and find ${\overline{X_2}}$. Then, we find the next estimate ${\widehat{\overline{X_2}}}$ to $\mu$ in the following way. And, this process continues.

\begin{gather*}
\widehat{\overline{X_1}}={\overline{X_1}} \\  
\widehat{\overline{X_2}}=\frac{{\overline{X_1}}+{\overline{X_2}}}{2}=\frac{1}{2}\widehat{\overline{X_1}}+\frac{1}{2}{\overline{X_2}}\\
\widehat{\overline{X_3}}=\frac{{\overline{X_1}}+{\overline{X_2}}+{\overline{X_3}}}{3}=\frac{2}{3}\frac{({\overline{X_1}}+{\overline{X_2}})}{2}+\frac{1}{3}{\overline{X_3}}=\frac{2}{3} \widehat{\overline{X_2}}+\frac{1}{3}{\overline{X_3}}\\
\dots\\
\widehat{\overline{X_k}}=\frac{k-1}{k}\widehat{\overline{X_{k-1}}}+\frac{1}{k}{\overline{X_k}}
\label{e5}
\end{gather*}

As the mean of all possible samples' means is the population mean, the series $({\widehat{\overline{X_k}}})_{k\in {\mathbb{N}}}$ approaches to $\mu$. For our study, population size is normally large. So, neither consideration of all possible samples nor finding of $\mu$ is possible. We, therefore, declare ${\widehat{\overline{X_k}}}$ as our final estimate to the population mean if $\frac {|\widehat{\overline{X_k}}-\widehat{\overline{X_{k-1}}}|}{\widehat{\overline{X_k}}} < \delta$, where $\delta$ is a small threshold value and specifies the precision we desire to achieve. We consider here $\delta = 0.01$.

Now, fixing of the `$m$' value is another important task of this calculation. Here, we use another statistical method for choosing $m$. For calculating the sample size ($m$), we first take a random sample of size $n_1$. Then, we find another sample size $n_2$ using the following equation~\cite{willam}.


\begin{equation}
\label{e1}
n_2=\frac{S_1^2}{C{\mu}_1^2}(1+8C+\frac{S_1^2}{n_1{\mu}_1^2}+\frac{2}{n_1})
\end{equation}
where $\mu_1$ and $S_1^2$ are the mean and variance of the first sample of size $n_1$, and 
\begin{equation}
\label{e2}
C=\frac{r^2}{t^2}
\end{equation}
where $t$ is the constant and $r$ is the relative error. For our experimental setup, we consider $t=2$ and $r=0.05$~\cite{willam}.

As a next step, we randomly and uniformly take the second sample of size $n_2$. Then, we find $\mu_2$ and $S^2_2$ as the mean and variance of the second sample. Using these parameters, we find another sample size $m_0$, which finally leads us to get the `m':

\begin{equation}
\label{e3}
m_0=\frac{t^2S_2^2}{r^2\mu_2^2}
\end{equation}

\noindent Now, the desired sample size is calculated as following, where $N$ is the population size.

\begin{equation}
\label{e4}
m=\frac{m_0}{1+\frac{m_0}{N}}
\end{equation}


\subsection{The Method of Experiment}
\label{method}
Though Algorithm~\ref{alg1} is a decision algorithm, a slight modification in Algorithm~\ref{alg1} enables us to get the total number of edges explored by it. To do that, we initialize a variable $Total\_count$ ($Total\_count \gets 0$) in the Step 2 of Algorithm~\ref{alg1}, and rewrite the Step 8 as following:

\fbox{\begin{minipage}{40em}
{\em \textbf{Step 8:}}\\ (a) Mark the labels ($l'_j$s) which are not in any path, computed in Step 7, as {\em irrelevant}.\\
(b) $Total\_count$ $\gets$ $Total\_count$ $+$ 2*Count.\\
(c) $Count\gets 2*Count$ $-$ $\#$ $irrelevant$ $labels$.
\end{minipage}}

So, we just add an extra step (Step 8(b)) in Algorithm 1 to get the number ($Total\_count$) of explored edges. We use this modified algorithm in our experimentation. However, Algorithm~\ref{alg1} demands two input parameters - one is a CA (that is, a rule vector) and the other is a pair of states (source and destination). For the experiment, therefore, we need to find out sample size twice. One for the pairs of states when a CA is given, and the other for the CAs of a given size. Let us consider that $m''$ be the number of CAs to be sampled for a given size, and $m'$ be the number of pairs to be sampled for a given CA.


\begin{example}
This example illustrates, the calculation of $m'$. Let us consider the 20-cell CA $\langle$106, 110, 191, 148, 71, 118, 189, 147, 164, 141, 90, 183, 201, 73, 106, 103, 230, 207, 73, 36$\rangle$. To find $m'$, we first randomly choose 500 ($=n_1$) pairs of source and destination states. By Algorithm~\ref{alg1}, we can calculate $\mu_1=359$ (mean of explored edges), $S^2_1=11025$. Using the values of mean and variance, we find $n_2=138$ (using Equation~\ref{e1}). For the sample size $n_2$, we get the $\mu_2=352$ and $S^2_2=9978$. Now using Equation~\ref{e3}, we get the value of $m_0=129$. Finally, we get the sample size $m'$, which is also 129 (using Equation~\ref{e4}) where $N=2^{20}$. 
\end{example}

Now, using the value of $m'$ and $m''$, we can get the average number of explored edges. For a given CA size, we randomly and uniformly synthesize $m''$ number of (non-uniform) CAs, and for each CA we randomly and uniformly choose $m'$ number of source and destination pairs. However, for each case, we use the modified Algorithm~\ref{alg1} to get total number of edges explored to decide the reachability. For ease of reference, the method is summarized in Algorithm~\ref{alg2}. This method takes the CA size as input, and reports the average number of explored edges. We use this algorithm to get our further results.


\begin{algorithm2e}[h]
\caption{Average Number of Explored Edges}
\KwIn{\textit{CA Size (n)}}
\KwOut{\textit {Average number of edges explored}}
{\em \textbf{Step 1:}} Set $pre\_avg\_edges$ $\gets$ 0 and $k \gets 1$\\
{\em \textbf{Step 2:}} Set $\#explored\_edges$ $\gets$ 0 and $i \gets 0$\\
{\em \textbf{Step 3:}} If $i\geq m''$, then goto Step 9. \\
{\em \textbf{Step 4:}} (a) Choose an $n$-cell CA randomly \\
 \hspace{3.5em} (b) Set $j \gets 0$\\
{\em \textbf{Step 5:}} If $j\geq m'$, then set $i \gets i+1$ and goto Step 3.\\
{\em \textbf{Step 6:}} Synthesize $n$-bit Source ($S$) and Destination ($D$) states randomly\\
{\em \textbf{Step 7:}} (a) With the help of the modified Algorithm~\ref{alg1}, get the total number\\
 \hspace{3.5em}  of explored edges for this CA and ($S, D$). Suppose, the number is\\
 \hspace{3.5em}  $Total\_count$.\\
 \hspace{3.5em} (b) Set $\#explored\_edges$ $\gets$ $\#explored\_edges$ $+$ $Total\_count$\\
{\em \textbf{Step 8:}} Set $j \gets j+1$ and goto Step 5.\\
{\em \textbf{Step 9:}} Set $avg\_edges$ $\gets$ $\frac{k-1}{k}$ $pre\_avg\_edges$ $+$ $\frac{1}{k}\frac{\#explored\_edges}{m'*m''}$ \\
{\em \textbf{Step 10:}} If $\frac {(avg\_edges) - (pre\_avg\_edges)}{avg\_edges} < \delta$, then\\
 \hspace{3.5em} report $\textbf{avg\_edges}$ and $\textbf{exit}$.\\
{\em \textbf{Step 11:}} Set $pre\_avg\_edges$ $\gets$ $avg\_edges$, $k \gets k + 1$ and goto Step 2.

\label{alg2}
\end{algorithm2e}


\subsection{The Results}
\label{result}
Using Algorithm~\ref{alg2}, we have extensively experimented with various CA sizes to get the average number of explored edges against a CA size. In Table~\ref{avg-table}, we report a sample experiment to show the average number of explored edges with respect to the size of automaton. The table points out the fact that with increase of CA size, explored number of edges also increases, but it is not exponential. 

\begin{table}[h]
	\setlength{\tabcolsep}{4.5pt}
	\begin{center}
		\caption{Experimental results}	
		\label{avg-table}
		\begin{tabular}{|c||c|c|c|c|c|c|c|c|c|c|}\hline
 		CA size & 10 & 20 & 30 & 40 & 50 & 60 & 70 & 80 & 90 & 100\\\hline
 		
 		Average number of &  &  &  &  &  &  &  &  &  & \\
 		edges to be & 50 & 344 & 1085 & 2428 & 4536 & 7612 & 11704 & 17012 & 23742 & 31923 \\
 		explore &  &  &  &  &  &  &  &  &  & \\\hline
 		
		
		\end{tabular}
	\end{center}
\end{table}
Therefore, we need another experiment for finding the {\em rate of growth} with respect to CA size. However, the worst case time complexity is exponential for this problem. Therefore, we can compare the average number of explored edges (experimentally) with the worst case of reachability problem for different size of automaton. In Figure~\ref{graph1}, we plot the logarithm of number of edges explored against the CA size. The worst case scenario is shown by the dotted line and experimental result is shown by continuous curve in the figure. It is obvious from the graph that the edges explored on average is much less than that on worst case. 

\begin{figure*}[h]
	\centering
	\includegraphics[width= 11.2in, height = 2.9in]{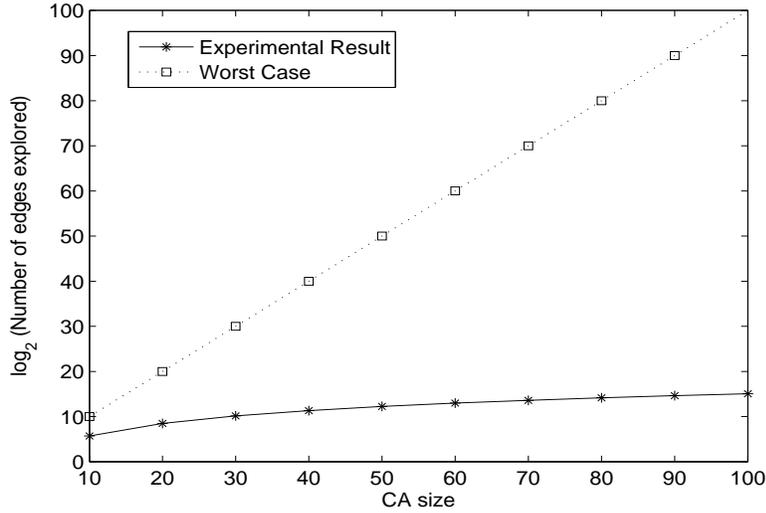}
	\caption{Average number of edges explored (experimentally) and worst case} 
	\label{graph1}
\end{figure*}

\subsection{The Rate of Growth}
\label{avg-per}
Experimental results indicate that the rate of growth of average number of explored edges is not exponential. In this sub-section, we find the rate of growth of explored edges to mathematically feel the change in explored edges with respect to the size of automaton. To find the rate of growth, we use the {\em empirical curve bounding technique}~\cite{Catherine}. Assuming the explored edges ($e$) follows power rule, that is, $e\approx kn^a$~\cite{Catherine}, the coefficient `$a$' can be found by taking empirical measurements of explored edges $\{e_1, e_2\}$ at some input CA size $\{n_1, n_2\}$, and calculating $\frac{e_2}{e_1}\approx(\frac{n_2}{n_1})^a$. So,
\begin{equation}
\label{rate}
  a\approx \frac{\log (e_2/e_1)}{\log(n_2/n_1)}
\end{equation}

Now, after taking the value of `$e$' for different size of automaton, we can find rate of growth using the Equation~\ref{rate}. In Table~\ref{table-rate}, we are showing the rate of growth with respect to CA size. 

\begin{table}[h]
	\setlength{\tabcolsep}{4.5pt}
	\begin{center}	
		\caption{Rate of growth with respect to CA size}
		\label{table-rate}
		\begin{tabular}{|c||c|c|c|c|c|c|c|c|c|c|}\hline
 		CA size & 10 & 20 & 30 & 40 & 50 & 60 & 70 & 80 & 90 & 100\\\hline
 		
 		Explored edges & 50 & 344 & 1085 & 2428 & 4536 & 7612 & 11704 & 17012 & 23742 & 31923 \\\hline
 		
 		Rate of growth ($a$) & - & 2.78 & 2.83 & 2.80 & 2.80 & 2.84 & 2.79 & 2.80 & 2.83 & 2.81\\\hline
		
		\end{tabular}
	\end{center}
\end{table}

From the experimentation, we have also observed that the growth rate of explored edges always lies under some upper bound. To represent this fact asymptotically, we are using the big-oh ($O$) notation. From the definition of big-oh ($O$) notation, we can get that for a given function $g(n)$, $T(n)=O(g(n))$, if there exist two positive constant $c$ and $n_0$, such that $0\leq T(n)\leq cg(n)$, for all $n\geq n_0$~\cite{Cormen}. As the average number of edges, explored of the non-uniform CAs satisfies the definition of big-oh, so we represent the rate of growth by big-oh notation. From Table~\ref{table-rate}, we can show that, the value of `$a$' is nearly $3$ for all value of $n$ ($n$ is the size of automaton). So, we estimate $g(n)=n^3$. Hence, we can say, the average number of edges to be explore of these CAs as $O(n^3)$. This is validated in Fig.~\ref{graph2}.

\begin{figure*}[h]
	\centering
	\includegraphics[width= 8.2in, height = 2.9in]{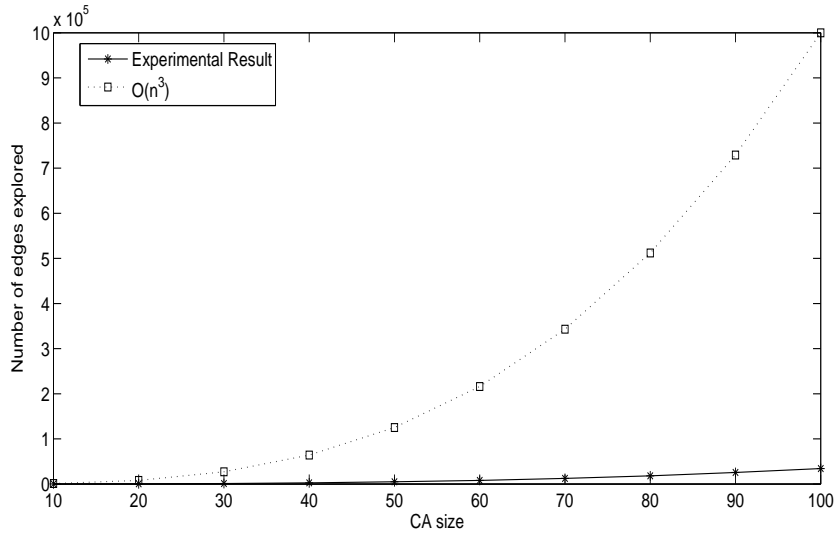}
	\caption{Upper bound of average number of edges explored}
	\label{graph2}
\end{figure*}

%
%

\section{Conclusion}
This paper has presented an in-depth analysis on the non-uniform CAs for reachability problem. The reachability tree has been utilized to develop theories for this class of CAs. We have introduced here a technique to trace the state transition diagram in reachability tree. This technique has helped us to design the decision algorithm for the reachability problem. The average case analysis of our algorithm is done experimentally. The average case performance is $O(n^3)$ of our algorithm, where the worst case time complexity is exponential.


\bibliographystyle{fundam}
\bibliography{References}

\end{document}